\documentclass[envcountsame]{llncs}

\usepackage{amsmath,amssymb}
\usepackage[dvips]{graphicx}
\usepackage{hyperref}

\begin{document}

\title{A Block-Sensitivity Lower Bound for Quantum Testing Hamming Distance}
\author{Marcos Villagra}
\institute{N\'ucleo de Investigaci\'on y Desarrollo Tecnol\'ogico\\
	Universidad Nacional de Asunci\'on, Paraguay
	\\
	\email{mvillagra@pol.una.py}}
\maketitle

\begin{abstract}
\noindent The Gap-Hamming distance problem is the promise problem of deciding if the Hamming distance $h$ between two strings of length $n$ is greater than $a$ or less than $b$, where the gap $g=|a-b|\geq 1$ and $a$ and $b$ could depend on $n$. In this short note, we give a lower bound of $\Omega( \sqrt{n/g})$ on the quantum query complexity of computing the Gap-Hamming distance between two given strings of lenght $n$. The proof is a combinatorial argument based on block sensitivity and a reduction from a threshold function.

\end{abstract}

\section{Introduction}

A generalized definition of the Hamming distance is the following: given two strings $x$ and $y$, decide if the Hamming distance $h(x,y)$ is greater than $a$ or less than $b$, with the condition that $b<a$. Note that this definition gives a partial boolean function for the Hamming distance with a gap. There is a entire body of work on the computation of a particular case of this notion of Hamming distance in the decision tree and communication models known as the {\it Gap-Hamming distance} (GHD) problem, which asks to differentiate the cases $h(x,y)\leq n/2-\sqrt{n}$ and $h(x,y)\geq n/2+\sqrt{n}$ \cite{Woo07}. A lower bound on GHD implies a lower bound on the memory requirements of computing the number of distinct elements in a data stream \cite{IW03}. Chakrabarti and Regev \cite{CR11} give a tight lower bound of $\Omega(n)$; their proof was later improved by Vidick \cite{Vid10} and then by Sherstov \cite{She11}. For the Hamming distance with a gap of the form $n/2\pm g$ for some given $g$, Chakrabarti and Regev also prove a tight lower bound of $\Omega(n^2/g^2)$. In the quantum setting, there is a communication protocol with cost $\mathcal{O}(\sqrt{n}\log n)$ \cite{BCW98}.

Suppose we are given oracle access to input strings $x$ and $y$. In this note, we prove a lower bound on the number of queries to a quantum oracle to compute the Gap-Hamming distance with an arbitrary gap, that is, for any given $g=a-b$.

\begin{theorem}\label{the:lower-bound}
Let $x,y\in\{0,1\}^n$ and $g=a-b$ with $0\leq b<a\leq n$. Any quantum query algorithm for deciding if $h(x,y)\geq a$ or $h(x,y)\leq b$ with bounded-error, with the promise that one of the cases hold, makes at least $\Omega( \sqrt{n/g} )$ quantum oracle queries.
\end{theorem}

The proof is a combinatorial argument based on block sensitivity. The key ingredient is a reduction from a a threshold function. A previous result of Nayak and Wu \cite{NW99} implies a tight lower bound of $\Omega(\sqrt{n/g}+\sqrt{h(n-h)}/g)$; their proof, however, is based on the polynomial method of Beals \emph{et al.} \cite{BBC01} and it is highly involved. The proof presented here, even though it is not tight, is simpler and requires no heavy machinery from the theory of polynomials.

\section{Proof of Theorem \ref{the:lower-bound}}

Let $a,b$ be such that $0\leq b <a\leq n$. Define the partial boolean function $GapThr_{a,b}$ on $\{0,1\}^n$ as
\begin{equation}
GapThr_{a,b}(x)=\left\{ \begin{matrix}
	1	& \text{ if } |x| \geq a\\
	0	& \text{ if } |x| \leq b.
\end{matrix} \right.
\end{equation}

To compute $GapThr_{a,b}$ for some input $x$, it suffices to compute the Hamming distance between $x$ and the all 0 string. Thus, a lower bound for Gap-Hamming distance  follows from a lower bound for $GapThr_{a,b}$.

Let $f:\{0,1\}^n\to \{0,1\}$ be a function, $x\in \{0,1\}^n$ and $B\subseteq \{1,\dots,n\}$ a set of indices called a block. Let $x^B$ denote the string obtained from $x$ by flipping the variables in $B$. We say that $f$ is \emph{sensitive} to $B$ on $x$ if $f(x)\neq f(x^B)$. The block sensitivity $bs_x(f)$ of $f$ on $x$ is the maximum number $t$ for which there exist $t$ disjoint sets of blocks $B_1,\dots,B_t$ such that $f$ is sensitive to each $B_i$ on $x$. The {\it block sensitivity} $bs(f)$ of $f$ is the maximum of $bs_x(f)$ over all $x\in\{0,1\}^n$.

From Beals \emph{et al.} \cite{BBC01} we know that the square root of block sensitivity is a lower bound on the bounded-error quantum query complexity. Thus, Theorem \ref{the:lower-bound} follows inmediately from the lemma below.

\begin{lemma}\label{lem:lower-bound}
$bs(GapThr_{a,b})=\Theta(n/g)$.
\end{lemma}
\begin{proof}
Let $x\in\{0,1\}^n$ be such that $GapThr_{a,b}(x)=0$ and suppose that $|x|=b$. To obtain a 1-output from $x$ we need to flip at least $g=a-b$ bits of $x$. Hence, we divide the $n-b$ least significant bits of $x$ in non-intersecting blocks, where each block flips exactly $g$ bits. The number of blocks is $\lfloor\frac{n-b}{a-b}\rfloor$, which is at most $bs_x(GapThr_{a,b})$. To see that $\lfloor\frac{n-b}{a-b}\rfloor$ is the maximum number of such non-intersecting blocks, consider what happens when the size of a block is different from $g$. If the size of a block is less that $g$, then we cannot obtain a 1-output from $x$; if the size of a block is greater than $g$, then the number of blocks decreases. Thus, we have that $bs_x(GapThr_{a,b})=\lfloor\frac{n-b}{g}\rfloor$.

For any $x'$ with $|x'|<b$, we need to flip $a-b$ bits plus $b-|x'|$ bits. Using our argument of the previous paragraph, the size of each block is thus $g+b-|x'|$, and hence, $bs_{x'}(GapThr_{a,b})=\lfloor\frac{n-|x'|}{g+b-|x'|}\rfloor$. Note that $bs_{x'}(GapThr_{a,b})\leq bs_{x}(GapThr_{a,b})$.

For the case when $GapThr_{a,b}(x)=1$ and $|x|=a$, to obtain a 0-output from $x$ we need to flip at least $g$ bits of $x$. Hence the same argument applies, and thus, $bs_x(GapThr_{a,b})=\lfloor\frac{n-a}{g}\rfloor$.

Taking the maximum between the cases when $|x|=b$ and $|x|=a$, we have that $bs(GapThr_{a,b})=\max\{(n-b)/g,(n-a)/g\}=\Theta(n/g)$. \qed
\end{proof}





\bibliographystyle{splncs03}
\bibliography{../../library}

\end{document}